\documentclass[preprint]{elsarticle}

\usepackage{booktabs} 
\usepackage{subcaption}
\usepackage{multirow}
\usepackage{amsmath}
\usepackage{fancyvrb}
\usepackage{color}

\usepackage[margin=1in]{geometry}

\usepackage{paralist, tabularx}

\usepackage{hyperref}

\usepackage{multirow}
\usepackage[numbers]{natbib}
\newtheorem{definition}{Definition}  
\newtheorem{proof}{Proof}    
\newtheorem{theorem}{Theorem}   
\usepackage{algorithm}
\usepackage{algorithmic}

\usepackage{makecell} 

\journal{Information Sciences}

\begin{document}

\begin{frontmatter}

\title{TKUS: Mining Top-$k$ High-Utility Sequential Patterns}


\author[1]{Chunkai Zhang}
\ead{ckzhang@hit.edu.cn}
\author[1]{Zilin Du}
\ead{yorickzilindu@gmail.com}
\author[2]{Wensheng Gan\corref{ca}} 
\ead{wsgan001@gmail.com}
\cortext[ca]{Corresponding author}

\author[3]{Philip S. Yu}
\ead{psyu@uic.edu}

\address[1]{Department of Computer Science and Technology, Harbin Institute of Technology (Shenzhen), Shenzhen 518055, China}
\address[2]{College of Cyber Security, Jinan University, Guangzhou 510632, China}
\address[3]{Department of Computer Science, University of Illinois at Chicago, IL 60616, USA}

\begin{abstract}
High-utility sequential pattern mining (HUSPM) has recently emerged as a focus of intense research interest. The main task of HUSPM is to find all subsequences, within a quantitative sequential database, that have high utility with respect to a user-defined minimum utility threshold. However, it is difficult to specify the minimum utility threshold, especially when database features, which are invisible in most cases, are not understood. To handle this problem, top-$k$ HUSPM was proposed. Up to now, only very preliminary work has been conducted to capture top-$k$ HUSPs, and existing strategies require improvement in terms of running time, memory consumption, unpromising candidate filtering, and scalability. Moreover, no systematic problem statement has been defined. In this paper, we formulate the problem of top-$k$ HUSPM and propose a novel algorithm called TKUS. To improve efficiency, TKUS adopts a projection and local search mechanism and employs several schemes, including the Sequence Utility Raising, Terminate Descendants Early, and Eliminate Unpromising Items strategies, which allow it to greatly reduce the search space. Finally, experimental results demonstrate that TKUS can achieve sufficiently good top-$k$ HUSPM performance compared to state-of-the-art algorithm TKHUS-Span.
\end{abstract}

\begin{keyword}
 utility mining \sep sequence data \sep pattern mining \sep top-$k$ \sep high-utility sequence
\end{keyword}

\end{frontmatter}

\section{Introduction}

We are currently in the age of big data. Sequential pattern mining (SPM), which has been very popular since it was first proposed \cite{agrawal1995mining} in the early 1990s, has been successfully applied to many realistic scenarios, such as bioinformatics \cite{wang2007frequent}, consumer behavior analysis \cite{srikant1996mining}, and webpage click-stream mining \cite{fournier2012using}. The goal of SPM is to extract all frequent sequences (as sequential patterns) from a sequence database with respect to a user-defined minimum threshold called "support". In recent years, multiple approaches \cite{fournier2017survey,han2001prefixspan} have been developed to achieve this goal by efficiently discovering patterns reflecting the potential connections within items.

In a frequency-oriented pattern mining framework \cite{fournier2017survey,han2001prefixspan}, where the frequency is the only metric for a sequence, many infrequent but crucial patterns are likely to be missed. In other words, most of the patterns selected by SPM algorithms are uninformative because the frequency of a pattern does not always fully correspond to its significance (i.e., profit, interest) \cite{truong2019survey}. Consider a recommendation system in an electronics store; clearly, the main task of the system is to generate more benefits for the business. As is well known, the unit profit from the sale of luxury goods, such as large screen OLED TV, is much greater than that of everyday supplies, such as batteries; however, the former are sold in much lower volumes than the latter. In this scenario, the system will tend to emphasize the patterns that lead to the purchase of luxury goods, yielding a higher profit for revenue maximization instead of the frequent purchases of everyday supplies. To deal with this problem in the frequency-oriented pattern mining framework, the concept of utility was incorporated and high-utility itemset mining (HUIM) \cite{gan2018survey,yao2004foundational} belonging to the utility-oriented framework, which considers the relative importance of items \cite{gan2018surveyuo}, was developed. HUIM has been widely researched. It has been found that although HUIM methods are able to gain valuable information in certain practical applications, they are incapable of addressing sequential databases, where each item has a timestamp. To deal with this problem, high-utility sequential pattern mining (HUSPM) \cite{shie2011mining,yin2012uspan,zhang2019two} was developed, and it has become an emerging topic of interest in the domain of knowledge discovery in databases (KDD) \cite{fayyad1996data,frawley1992knowledge}.

\begin{figure}[ht]
	\centering
	\includegraphics[clip,scale=0.4]{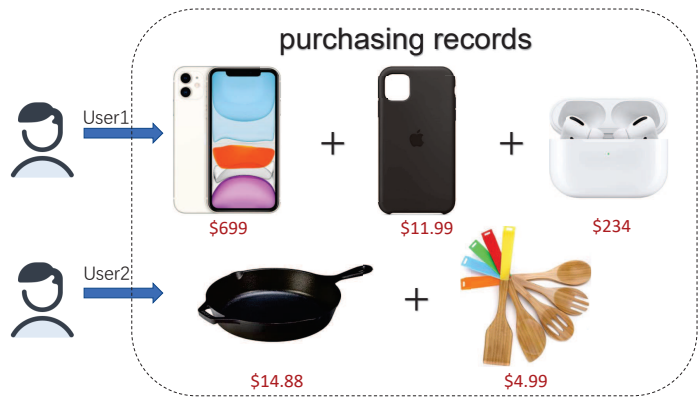}
	\caption{Example of purchasing records in a consumer retail store}
	\label{retail}
\end{figure}

HUSPM can be applied in many common scenarios. Consider the real-life example of a consumer retail store as shown in Figure \ref{retail}, where users purchase a series of items, each of which has a corresponding unit profit, at different times. These consumers' purchasing behaviors make up a large-scale transaction database containing a large amount of underlying knowledge, which can be discovered by HUSPM, for decision-making. In contrast to other pattern mining tasks, HUSPM emerged quickly and had attracted much attention. To represent the relative importance of patterns in the HUSPM problem, each item in the mining object, called a quantitative sequence database, is associated with a positive value called internal utility, which generally represents the number of occurrences (e.g., the number of the items purchased by a customer in one visit). Moreover, each kind of item appearing in the database is associated with a special value called external utility, which indicates the item's relative importance (e.g., the unit profit of the item). The main task of HUSPM is to find all subsequences with high utility, i.e., high-utility sequential patterns (HUSPs), in a quantitative sequential database with respect to a user-defined minimum utility threshold. Specifically, compared to the aforementioned SPM and HUIM, HUSPM considers the chronological ordering of items as well as utility values associated with them \cite{gan2020fast}, which makes it a much more challenging and complex problem. Many researchers have proposed algorithms \cite{gan2020proum,zhang2019two,zhang2019efficient} with novel pruning strategies and data structures to efficiently mine HUSPs.

However, HUSPM has a limitation in identifying HUSPs that contain valuable information: it is difficult for users to specify the minimum utility threshold, especially when they are not familiar with the database's features, such as the average number of items per itemset, total number of sequences, and distribution of utilities, which are customarily invisible to the user. For example, if the threshold is set to a too small value, we may discover too many HUSPs with unimportant and redundant information, whereas if the threshold is too large,  we may capture only a few HUSPs that cannot provide sufficient information. Furthermore, given the same threshold, one database may yield millions of HUSPs, whereas another may yield none. It is time-consuming for decision makers to fine-tune the threshold to extract the proper number of patterns for their intended purpose. An interesting but difficult question arises from this: how do we mine an appropriate number of HUSPs? Top-$k$ HUSPM was proposed to answer this question \cite{yin2013efficiently}. For instance, Figure \ref{retail} can be considered as a motivated real-world application of top-$k$ utility mining. In general, users are more inclined to find out the top-$k$ profitable products instead of hundreds of thousands of results in retail store. There is a need to analyze sales data to establish sales strategies related to retail benefits such as inventory preparation, product arrangement, and promotion. If there are three available promotion positions on the shelf, then top-3 HUSPM technology can be used to discover the three patterns with the highest utility. Assume one of the patterns is $<$\{\textit{ham} \textit{cheese}\}, \{\textit{milk}\}$>$, decision-makers can put \textit{ham} and \textit{cheese} on sale and then arrange \textit{milk} into the promotion position for cross-marketing based on the mining results. Furthermore, it is necessary to perform a rapid analysis with respect to huge sales databases during non-opening hours for their smooth running according to plan \cite{ryang2015top}. With no requirement for a predefined minimum utility threshold, top-$k$ HUSPM selects patterns with the top-$k$ highest utilities; it is inspired by top-$k$ SPM \cite{tzvetkov2005tsp} and top-$k$ HUIM \cite{wu2012mining}. Different from classic HUSPM, top-$k$ HUSPM is simpler and more user-friendly because setting the value of $k$, the number of desired HUSPs, is more straightforward than specifying the minimum utility threshold which requires domain knowledge.

In practice, top-$k$ HUSPM can play a significant role in many real-life applications, such as web log mining \cite{ahmed2010mining}, gene regulation analysis in bioinformatics \cite{zihayat2016top}, mobile computing \cite{shie2011mining}, and cross-marketing in retail stores \cite{yin2013efficiently}. Some typical challenges faced by top-$k$ HUSPM are listed here.

\begin{itemize}
	\item
	Frequency of patterns in a frequency-oriented framework are monotone, but the download closure property is not held in the utility-oriented framework. Therefore, it is computationally infeasible to reduce the search space with existing SPM pruning strategies.
	
	\item 	Compared to top-$k$ HUIM, top-$k$ HUSPM is intrinsically better equipped to face a critical combinatorial explosion of the search space. This is because the sequential ordering of items leads to various possibilities of concatenation in a quantitative sequential database. This means that HUSPM must check more candidates than HUIM, which can lead to high computational complexity without powerful pruning strategies.
	
	\item	With the purpose of guaranteeing algorithm completeness, i.e., missing no top-$k$ pattern, the minimum utility threshold must be increased from a very low value (very close or equal to zero) because the threshold is not specified in advance. Increasing the threshold as fast as possible with efficient strategies is very challenging in top-$k$ HUSPM.
\end{itemize}

Up to date, only very preliminary work \cite{wang2016efficiently, yin2013efficiently, zihayat2016top} has been conducted to capture top-$k$ HUSPs. The research topic is still at a very early stage of development, and existing strategies require significant improvements in terms of running time, memory consumption, unpromising candidates filtering, and scalability. Moreover, no systematic problem statement is defined. Therefore, in this paper, we formulate the problem of top-$k$ HUSPM and propose a novel algorithm called TKUS. The major contributions of this study can be summarized as follows:

\begin{itemize}
	\item 	We address the concept of top-$k$ HUSP by considering not only frequency but also utility value. We also formulate the problem of top-$k$ HUSPM. In particular, important notations and concepts of top-$k$ HUSPM are defined.
	
	\item	With the purpose of overcoming the challenges mentioned earlier, a novel algorithm called TKUS is proposed. To ensure that all top-$k$ HUSPs are found, we investigate the Sequence Utility Raising (SUR) strategy to increase the minimum utility threshold quickly. For further efficiency improvement, we adopt two utility upper bounds and design two companion pruning strategies: Terminate Descendants Early (TDE) and Eliminate Unpromising Items (EUI).
	
	\item 	Extensive experiments using various algorithms on both real-world and synthetic datasets demonstrate that the proposed TKUS has excellent performance in terms of runtime, memory usage, unpromising candidate filtering, and scalability. In particular, experimental results comparing our proposed scheme with state-of-the-art algorithm TKHUS-Span demonstrate that TKUS achieves sufficiently good performance for top-$k$ HUSPM compared to TKHUS-Span.
\end{itemize}

The remainder of this paper is organized as follows. Related work is briefly reviewed in Section 2. Then, the top-$k$ HUSPM problem is formulated in Section 3 alongside related definitions. The proposed TKUS algorithm is presented with several strategies and data structures in Section 4. Experimental results are presented and evaluated in Section 5. Finally, Section 6 concludes the paper, and future work is also discussed.

\section{Related Work}

All research work related to this topic can be divided into three categories: (1) SPM algorithms, (2) HUSPM algorithms, and (3) key algorithms for mining top-k high-utility patterns. 

\subsection{Sequential Pattern Mining}

Agrawal and Srikant \cite{agrawal1995mining} presented the first definition of SPM when considering customer consumption records. They also presented a simple algorithm called AprioriAll based on the Apriori property \cite{agrawal1994fast}. Resembling AprioriAll in its mining principle, GSP, which greatly outperforms AprioriAll, was introduced by Srikant and Agrawal \cite{srikant1996mining}. However, the drawback of GSP is that it traverses the original database repeatedly to calculate the support of candidate patterns, which incurs very high computational costs. Later,  an alternative algorithm called SPADE \cite{zaki2001spade} with a vertical database representation is developed to resolve the repetitive scanning problem. SPADE is able to decompose the original search space into smaller pieces that are independently solved according to combinatorial properties; this efficiently reduces the amount of scanning required. The excellent search schemes result in only three database scans, or even one single scan in the case of SPADE, which minimizes the I/O costs. Similar to SPADE in adopting a vertical database, SPAM \cite{ayres2002sequential} can extract very long sequential patterns with a novel depth-first search strategy. As is well known, algorithms will generate a large number of candidates when handling dense datasets because of combinatorial explosion, making them ineffective. Therefore, Yang \textit{et al.} \cite{yang2007lapin} developed LAPIN using a straightforward but innovative idea that the last occurrence position of an item determines whether to continue concatenating candidates or not. All of the aforementioned algorithms can be considered Apriori-based algorithms.

Note that these algorithms have a disadvantage of generating many unpromising candidates that have no possibility of appearing in the database, similar to other Apriori-based algorithms. To solve this problem, a series of pattern growth algorithms have been proposed. For instance, the high-efficiency FreeSpan, which was developed by Han \textit{et al.} \cite{han2000freespan}, adopts projected sequential databases built recursively by frequent items to grow pattern fragments. For further improvement, they proposed an improved version called PrefixSpan \cite{han2001prefixspan} with two projection strategies called level-by-level projection and bi-level projection. PrefixSpan projects only corresponding postfix subsequences into the projected database, which can greatly decrease the scope of scanning, allowing it to run fast, especially when the desired sequential patterns are numerous and/or long. As a parallelized version of PrefixSpan utilizing MapReduce, Sequence-Growth \cite{liang2015sequence} adopts a lexicographical order to construct candidates with a breadth-wide support-based approach called lazy mining. Moreover, some efficient data structures, such as Web access pattern tree (WAP-tree) \cite{pei2000mining} and its preorder linked and position coded version \cite{ezeife2005plwap}, have been proposed to compress databases and improve efficiency.

However, pattern growth SPM algorithms have the crucial limitation that recursively building projected databases incurs high computational costs. Consequently, several early pruning strategies have been designed to avoid constructing projected storing structures of unpromising sequences. For example, DISC-all \cite{chiu2004efficient} adopts a novel pruning strategy called DISC to remove unpromising patterns early based on sequences of the same length. Subsequently, Chen \cite{chen2009updown} proposed a novel data structure called UpDown Directed Acyclic Graph (UDDAG), which results in fewer levels of recursion and faster pattern growth when constructing projected databases. UDDAG scales up much better than PrefixSpan and can be extended to applications with large search spaces. More details on SPM can be found in literature reviews \cite{fournier2017survey,gan2019survey}.

\subsection{High-Utility Sequential Pattern Mining}

The most important metric of SPM algorithms is frequency. However, frequency does not directly correspond to significance under any circumstances. To address this problem, SPM was generalized to HUSPM \cite{ahmed2010novel}, whose goal is to find all HUSPs in a quantitative sequential database with respect to a user-defined minimum utility threshold. With not only a utility attached to each item but also a timestamp, HUSPM has played a key role in various applications \cite{ahmed2010mining,shie2011mining,zihayat2016top}. Up to now, multiple HUSPM algorithms have been developed \cite{ahmed2010novel,shie2011mining,yin2012uspan,zhang2019two}, and it has become easy to obtain HUSPs using various optimization methods, such as efficient pruning strategies and highly compressed data representations. Ahmed \textit{et al.} \cite{ahmed2010novel}  were the first to incorporate the concept of utility into SPM. Along with defining the problem, they designed a new mining framework to find a complete set of HUSPs, where both internal and external utilities are considered. Furthermore, they proposed two two-phase algorithms called Utility Level (UL) \cite{ahmed2010novel} and Utility Span (US) \cite{ahmed2010novel} based on candidate generation and pattern growth approaches, respectively, adopting an upper bound called \textit{SWU}. Compared to the straightforward UL algorithm, US generates no candidates in the mining process.

There are two main drawbacks to the aforementioned algorithms. One is that algorithms generate many sequences with high \textit{SWU} values, which consumes considerable main memory in the first phase. The other is that scanning the database to calculate the utility of candidates incurs very high computational costs. To address these problems, subsequent studies \cite{yin2012uspan,wang2016efficiently} have adopted a prefixed tree, where each node denotes a candidate sequence (except the root, which is null), and its child nodes can be extended from it by one extension operation. One exception to this is the one-phase algorithm UM-Span \cite{shie2011mining}, which is applied in the real-life situation of planning mobile commerce environments. Due to the high complexity of sequential mobile transactions, tree-based algorithms have poor performance in such a scenario because they must construct a complex tree structure. UM-Span improves efficiency and can overcome the bottleneck of utility mining because it avoids additional scans to identify HUSPs by using a projected database-based approach.

All aforementioned HUSPM algorithms adopt the upper bound \textit{SWU}, which is very loose and still generates a large number of candidate sequences. In view of this, Yin \textit{et al.} \cite{yin2012uspan} designed a generic pattern selection framework and introduced an efficient algorithm called USpan, which utilizes two pruning strategies (width and depth) and a tree-based structure called the lexicographic $q$-sequence tree (LQS-Tree) to represent the search space. The width pruning strategy uses \textit{SWU} to remove unpromising items, whereas the depth pruning strategy uses  the \textit{SPU} upper bound to stop USpan from going deeper by identifying tree nodes. However, USpan may miss some HUSPs because \textit{SPU} sometimes filters promising candidates. The HuspExt \cite{alkan2015crom} and HUS-Span \cite{wang2016efficiently} algorithms are extensions of USpan devised to improve the mining efficiency. The HuspExt algorithm adopts a tight upper bound \textit{CRoM} to eliminate candidate items early. In HUS-Span, Wang \textit{et al.} \cite{wang2016efficiently} designed two tighter utility upper bounds, \textit{PEU} and \textit{RSU}, to remove unpromising patterns early and significantly reduce the search space. Recently, for further performance improvement, several algorithms have been proposed for HUSPM. For example, HUS-UT \cite{zhang2019efficient} adopts an efficient data structure called a utility table to facilitate the utility calculation; a parallel version called HUS-Par was also proposed. In addition, the novel data structures of utility-array and UL-list were proposed for ProUM \cite{gan2020proum} and HUSP-ULL \cite{gan2020fast}, respectively, to quickly discover HUSPs. 

Currently, there is an emergence of interesting extensions generalized from HUSPM. For example, Dinh \textit{et al.} \cite{dinh2017mining} designed a post-processing algorithm called PHUSPM, which can discover HUSPs periodically appearing in a quantitative sequential database. Adopting the special PBS and TSWU strategies, an efficient algorithm called MHUH  \cite{zhang2020efficient} was proposed for extracting high-utility hierarchical sequential patterns. In the domain of privacy preservation, Zhang \textit{et al.} \cite{zhang2019fast} developed a hiding HUSPs algorithm called FH-HUSP, which can protect personal private data based on dynamic programming and several efficient strategies. More recent HUSPM works are referenced in related surveys \cite{gan2018surveyuo,truong2019survey}.

\subsection{Top-$k$ Utility Pattern Mining}

Although aforementioned algorithms can efficiently find patterns, it is difficult for users to specify a proper minimum utility threshold. Top-$k$-based algorithms \cite{tzvetkov2005tsp,wang2005tfp,wang2016efficiently,wu2012mining} address this problem by providing users an opportunity to determine the desired number of patterns directly rather than considering the threshold.

Wu \textit{et al.} \cite{wu2012mining} first proposed an top-$k$ HUIM algorithm  for mining top-$k$ high-utility itemsets without setting the minimum utility threshold. They also incorporated several novel strategies for pruning the search space to achieve high efficiency. The study inspired a lot of works focusing on top-$k$ HUIM. Heungmo \textit{et al.} \cite{ryang2015top} developed an efficient algorithm REPT with highly decreased candidates. For further efficiency, Vincent \textit{et al.} \cite{tseng2015efficient} designed the first two-phase algorithm TKU with five strategies and the first one-phase algorithm TKO integrating the novel strategies. After that, kHMC \cite{duong2016efficient} relying on a novel co-occurrence pruning technique named EUCPT to avoid performing costly join operations for calculating the utilities of itemsets was designed. Moreover, several extension problem were proposed, such as discovering top-$k$ HUIs over data streams \cite{zihayat2014mining} and identifying top-$k$ on-shelf HUIs \cite{dam2017efficient}.

There are only a few top-$k$ HUSPM algorithms that dealing with sequence data.  TUS \cite{yin2013efficiently} is the first algorithm for top-$k$ HUSPM, and it discovers patterns without the minimum threshold by extending on their preliminary work of USpan. With the purpose of raising the minimum utility threshold quickly to reduce the search space as much as possible, TUS adopts three strategies in different stages of the mining process. As with the USpan algorithm \cite{yin2012uspan} utilizing the \textit{SPU} upper bound and corresponding pruning strategy, TUS is clearly also an incomplete algorithm, which means that it may miss some top-$k$ HUSPs under some circumstances. Besides, Wang \textit{et al.} \cite{wang2016efficiently} further developed the TKHUS-Span algorithm, which has three versions: \textit{BFS} strategy-based, \textit{DFS} strategy-based, and hybrid search strategy-based, based on the HUS-Span algorithm. TKHUS-Span adopting the \textit{BFS} strategy has better performance than other algorithms, including TUS, whereas that with the hybrid search strategy performs the best with limited memory space.

Applying top-$k$ HUSPM to a realistic scenario, Zihayat \textit{et al.} \cite{zihayat2016top} formulated a new problem as an extension of top-$k$ HUSPM: extracting the top-$k$ gene regulation-related patterns over time from a microarray dataset. They developed a novel utility model referring to a series of a priori professional knowledge on a specific disease from a biological investigation. They also designed a novel and problem-specific algorithm called TU-SEQ for mining the top-$k$ high-utility gene regulation sequential patterns with the ItemUtilList vertical data structure and PES strategy. The development of top-$k$ HUSPM is not yet mature. In particular, the only complete top-$k$ HUSPM method without missing any top-$k$ HUSPs, TKHUS-Span \cite{wang2016efficiently}, has much room for improvement, particularly in terms of efficiency. This motivates us to develop the more suitable data structure and more effective pruning strategies to address the problem of top-$k$ HUSPM.

\section{Preliminaries and Problem Formulation}

In this section, we briefly introduce the basic definitions and principles required for the remainder of the paper. We also adopt some definitions from prior research for clearer expression of the research issue. Finally, the problem definition of top-$k$ HUSPM is formalized.

\subsection{Notations and Concepts}

\begin{table}[!htbp]
	\centering
	\caption{Example quantitative sequence database}
	\label{table1}
	\begin{tabular}{|c|c|}  
		\hline 
		\textbf{SID} & \textbf{Quantitative sequence} \\
		\hline  
		\(S_{1}\) & $<$\{(\textit{a}:2) (\textit{c}:3)\}, \{(\textit{a}:3) (\textit{b}:1) (\textit{c}:2)\}, \{(\textit{e}:3)\}$>$ \\ 
		\hline
		\(S_{2}\) & $<$\{(\textit{a}:3) (\textit{d}:2)\}, \{(\textit{a}:1) (\textit{e}:3)\}, \{(\textit{b}:5) (\textit{c}:2) (\textit{d}:1) (\textit{e}:1)\}, \{(\textit{b}:1) (\textit{d}:5)\}$>$ \\  
		\hline  
		\(S_{3}\) & $<$\{(\textit{d}:2) (\textit{e}:2)\}, \{(\textit{a}:1) (\textit{b}:3)\}, \{(\textit{a}:2) (\textit{d}:4) (\textit{e}:1)\}, \{(\textit{f}:1)\}$>$ \\
		\hline  
		\(S_{4}\) & $<$\{(\textit{f}:2)\}, \{(\textit{a}:1) (\textit{d}:3)\} \{(\textit{d}:2)\} \{(\textit{a}:2)\}$>$ \\
		\hline
	\end{tabular}
\end{table}

\begin{table}[!htbp]
	\caption{Example utility table}
	\label{table2}
	\centering
	\begin{tabular}{|c|c|c|c|c|c|c|}
		\hline
		\textbf{Item}	    & \textit{a}	& \textit{b}	& \textit{c}	& \textit{d}	& \textit{e}	& \textit{f} \\ \hline 
		\textbf{External utility}	& \$5 & \$4	& \$2 & \$1 & \$3 & \$5 \\ \hline
	\end{tabular}
\end{table}

Let $I$ = \{$i_{1}$, $i_{2}$, $\cdots$, $i_{N}$\} be a set of possibly appearing and distinct items. An itemset \(X\) is a nonempty set containing one or more items of \(I\), that is, \(X \subseteq I\). The size of itemset \(X\) is represented by \(|X|\). A sequence $S$ = $<$$X_{1}$, $X_{2}$, $\cdots$, $X_{n}$$>$ is an ordered list of itemsets. Note that each element (i.e., itemset) can be unordered and satisfies \(X_{k} \subseteq I\), where \(1 \leq k \leq n\). Without loss of generality, all items in one itemset are sorted alphabetically. The length of \(S\) is \(l\) = \(\sum_{k = 1}^{n}|X_{k}|\), called an $l$-sequence, and the size of \(S\) is \(n\). We say that $T$: $<$$Y_{1}$, $Y_{2}$, ..., $Y_{m}$$>$ is the subsequence of \(S\), denoted as $T \subseteq S$, if there exists \(m\) integers \(1 \leq k_{1} < k_{2} < ... < k_{m} \leq n\) such that \(\forall 1 \leq v \leq m,Y_{v} \subseteq X_{k_{v}}\). For example, a sequence $s$ = $<$\{\textit{c}\}, \{\textit{a} \textit{b}\}$>$ is the subsequence of $<$\{\textit{a} \textit{c}\}, \{\textit{a} \textit{b} \textit{c} \textit{d}\}, \{\textit{b}\}$>$, and \(s\) is called a 3-sequence because its length and size are three and two, respectively. 

To illustrate the following concepts, we show examples of them in Tables \ref{table1} and \ref{table2}.

\begin{definition}[quantitative sequence]
	\rm A quantitative item ($q$-item) is a tuple ($i$:$q$) consisting of an item $i$ (\(i \in I\)) and a positive number $q$ representing the internal utility value (e.g., quantity) of $i$. A quantitative itemset ($q$-itemset) with $n$ $q$-items is denoted as \{($i_{1}$:$q_{1}$) ($i_{2}$:$q_{2}$)$\cdots$($i_{n}$:$q_{n})$\}, which can be regarded as an itemset with quantities. A quantitative sequence ($q$-sequence), denoted as $<$$ Y_{1}$, $Y_{2}$, ..., $Y_{m}$$>$, is an ordered list of $m$ $q$-itemsets, where $Y_{i} \subseteq I$.
\end{definition}

For example, $<$\{\textit{c}\}, \{\textit{a} \textit{b}\}$>$ is a sequence, while $<$\{(\textit{e}:6)\}, \{(\textit{f}:1) (\textit{c}:3)\}$>$ is a $q$-sequence, where each item is assigned a quantity.

\begin{definition}[quantitative sequence database]
	\rm A $q$-sequence database $D$ contains a set of pairs (\textit{SID}, \textit{QS}), where \textit{SID} is the unique identifier of a $q$-sequence \textit{QS}. Moreover, each kind of item in $D$ is associated with an external utility value (e.g., profit), all of which compose a utility table.
\end{definition}

Consider the example shown in Table \ref{table1}, where the $q$-sequence database has four $q$-sequences and six kinds of items, whose external utility values are provided in Table \ref{table2}. As a rule, the utility table is generated by decision makers according to prior knowledge of analogous users or contents.

\subsection{Utility Calculation}

In this subsection, we define a series of utility calculation functions. The $q$-item utility, denoted as \(q(i,\ j,\ s)\), which is the quantitative measure for ($i$:$q$) within the $j$th $q$-itemset of a $q$-sequence $s$, is defined as $u(i,\ j,\ s)$ = $q(i,\ j,\ s)$ $\times$ $eu(i)$, where \(q(i,\ j,\ s)\) and \(eu(i)\) are the corresponding internal utility in this occurrence and external utility of the item $i$, respectively. The $q$-itemset utility is equal to the sum of the utilities of the $q$-items it contains. Analogously, the utility of a $q$-sequence ($q$-sequence database) is the sum of the utilities of the $q$-itemsets ($q$-sequences) it contains.

For instance, consider the $q$-item $a$ within the 1st $q$-itemset of \(S_{1}\) in Table \ref{table1}; its utility can be calculated as $q(a,\ 1,\ S_{1})$ $\times$ $eu(a)$ = 2 $\times$ \$5 = \$10. Then, we have that the utility of the 1st $q$-itemset of \(S_{1}\) is \$16 (\$10 + \$6), and the utility of the $q$-sequence \(S_{1}\), described as \(u(S_{1})\), is \$48 (= \$16 + \$23 + \$9); the overall utility of the $q$-sequence database \(D\) in Table \ref{table1} is \(u(D)\) = \$48 + \$68 + \$42 + \$30 = \$188.

\begin{definition}[match]
	\rm Given an itemset $X$: $\{ i_{1},i_{2},\cdots,i_{m}\}$ and $q$-itemset $Y$: \{$j_{1}$:$q_{1}$) ($j_{2}$:$q_{2}$)$\cdots$($j_{n}$:$q_{n}$)\}, we say that $X$ matches $Y$ if and only if  $m$ = $n$ and $i_{k}$ = $j_{k}$ for \(1 \leq k \leq n\). Similarly, given a sequence $S$: $<$$X_{1}$, $X_{2}$, $\cdots$, $X_{m}$$>$ and $q$-sequence $Q$: $<$$Y_{1}$, $Y_{2}$, ..., $Y_{n}$$>$, $S$ matches $Q$, denoted as \(S\sim Q\) if and only if $m$ = $n$ and \(X_{k}\) matches \(Y_{k}\), where \(1 \leq k \leq n\).
\end{definition}

For example, \{\textit{a} \textit{c}\} matches the 1st $q$-itemset of \(S_{1}\), and $<$\{\textit{a} \textit{c}\}, \{\textit{a} \textit{b} \textit{c}\}, \{\textit{e}\}$>$ matches \(S_{1}\).

\begin{definition}[contain]
	\rm Given itemsets ($q$-itemsets) $X$ and $Y$, we say that $Y$ contains $X$ (i.e., $X$ is contained in $Y$), denoted as \(X \sqsubseteq Y\), if and only if $X$ is a subset of $Y$. Strictly speaking, the concept of "contain" slightly differs between an itemset and $q$-itemset.
\end{definition}

For example, \{\textit{a} \textit{b}\} is contained in \{\textit{a} \textit{b} \textit{c}\}, whereas \{(\textit{a}:1) (\textit{e}:2)\} is contained in \{(\textit{a}:1) (\textit{b}:3) (\textit{e}:2)\} and \{(\textit{a}:1) (\textit{e}:2) (\textit{f}:2)\} but not in \{(\textit{a}:3) (\textit{b}:2) (\textit{e}:2)\}. We formalize the calculation of itemset utility as $u(X,\ j,\ s)$ = $\sum_{\forall i \in X}^{}{u(i,\ j,\ s)}$, where $X$ is contained in the $j$th itemset of a $q$-sequence $s$. For example, in Table \ref{table1}, $u($\{\textit{a} \textit{c}\}, 2, $S_{1})$ = $u$(\textit{a}, 2, $S_{1}$) + $u$(\textit{c}, 2, $S_{1}$) = \$15 + \$4 = \$19.

\begin{definition}[instance]
	\rm Given a sequence $S$: $<$$X_{1},X_{2},\cdots,X_{m}$$>$ and $q$-sequence $Q$: $<$$Y_{1},Y_{2},\cdots,Y_{n}$$>$, where \(m \leq n\), we say that $Q$ has an instance $S$, denoted as \(S \sqsubseteq Q\) at position $p$: \(<k_{1},k_{2},...,k_{m}>\) if and only if there exists an integer sequence \(1 \leq k_{1} < k_{2} < ... < k_{m} \leq n\) such that \(X_{v} \sim X' \land X' \sqsubseteq Y_{k_{v}}\) \(\forall 1 \leq v \leq m\), where $X'$ is a $q$-sequence.
\end{definition}

For example, \(S_{2}\) has an instance $<$\{\textit{a}\}, \{\textit{e}\}$>$ at positions $p_{1}$: $<$1, 3$>$ and $p_{2}$: $<$2, 3$>$. In the remainder of this paper, for the sake of convenience, \(S \sqsubseteq Q\) is used to indicate \(S \sim S' \land S' \subseteq Q\).

\begin{definition}[sequence utility function]
	\rm In addition, the utility of a sequence $t$: $<$$X_{1}$, $X_{2}$, $\cdots$, $X_{m}$$>$ in a $q$-sequence $s$: $<$$Y_{1}$, $Y_{2}$, $\cdots$, $Y_{n}$$>$ at position $p$: $<$$k_{1}$, $k_{2}$, $\cdots$, $k_{m}$$>$ can be represented as \(u(t,p,s)\), where $u(t,p,s)$ = $\sum_{v = 1}^{m}u(X_{v},\ k_{v},\ s)$. In some cases, a sequence $t$ may appear in $s$ more than once; the set of all positions of $t$  is denoted as \(P(t,s)\). Let \(u(t,s)\) denote the utility of $t$ in $s$, and choose the maximum value as the utility, which is defined as $u(t,s)$ = $max\{u(t,\ p,\ s)|p \in P(t,s)\}$. The utility of $t$ in a $q$-sequence database $D$ can be defined as $u(t)$ = $\sum_{s \in D \land t \sqsubseteq s}^{}u(t,s)$. 
\end{definition}

For example, in Table \ref{table1}, \(S_{1}\) has an instance $<$\{\textit{a}\}, \{\textit{b} \textit{c}\}$>$ at position: $<$1, 2$>$; then, we can calculate $u$($<$\{\textit{a}\}, \{\textit{b} \textit{c}\}$>$, $<$1, 2$>$, $S_{1})$ = \$10 + \$8 = \$18. For example, in Table \ref{table1}, \(S_{2}\) has four instances of $t$: $<$\{\textit{a}\}, \{\textit{b}\}$>$, \(P(t,S_{2})\) = \{$<$1, 3$>$, $<$1, 4$>$, $<$2, 3$>$, $<$2, 4$>$\}; $u$($t$, $<$1, 3$>$, $S_{1}$) = \$15 + \$20 = \$35, \(u(t,S_{1})\) = max\{\$35, \$19, \$25, \$9\} = \$35, and \(u(t)\) = \$14 + \$35 + \$0 + \$0 = \$49. More details on how to calculate the utility value can be found in Ref. \cite{gan2018surveyuo}.

\begin{definition}[top-$k$ high-utility sequential pattern]
	\rm In a $q$-sequence database $D$, we define a sequence $s$ as a top-$k$ HUSP if there exist less than $k$ sequences whose utility values are higher than that of $s$.
\end{definition}

\textbf{Problem statement}. Given a $q$-sequence database $D$ and user-defined parameter $k$ in advance, the goal of top-$k$ HUSPM is to identify the complete set of top-$k$ HUSPs.

Conventional HUSPM aims to find the patterns whose utilities are no less than the minimum utility threshold \textit{minutil} in database $D$. Clearly, we can conclude that top-$k$ HUSPM, where the minimum utility threshold \textit{minutil} can be represented as \textit{minutil} = \textit{min}$\{u(t)|t \in T\}$, where $T$ is the complete set of top-$k$ HUSPs instead of being specified by the user, is an extension of conventional HUSPM.

\section{Proposed TKUS Algorithm}
Based on the aforementioned concepts, we developed an algorithm called TKUS for efficiently discovering top-$k$ HUSPs without specifying a minimum utility threshold. Without loss of generality, we formalize the following definitions and theorems under the context of a number of desired patterns $k$, a $q$-sequence database \(D\) including a series of $q$-sequences, and a utility table.

\subsection{Projection and Local Search Mechanism}

One of the main problems in the domain of data mining is that fundamental technologies repeatedly scan the database to discover knowledge. This issue also exists in HUSPM, which requires multiple database scans to calculate the utility of candidates. This can incur very high costs, especially for large database, even if some optimizations (e.g., Apriori property \cite{agrawal1994fast}) are performed to reduce the cost. To address this problem, we design a projection and local search mechanism, which only scans the original database once, and recursively constructs projected databases. This greatly limits the size of the scan space, especially when calculating the utility value of a long candidate, which contains so many items. To facilitate the presentation of algorithm, we present the following essential definitions.

\begin{definition}[extension]
	\rm The common operation of "extension", which includes $I$-Extension and $S$-Extension, is performed by many pattern mining algorithms. Given an $l$-sequence $t$, the $I$-Extension of sequence $t$ involves appending an item $i$ belonging to $I$ to the last itemset of $t$, which becomes a ($l$+1)-sequence \(t'\), also denoted as $<$$ t\bigoplus i $$>$. Similarly, the $S$-Extension of sequence $t$ involves appending the itemset only consisting of the item $i$ at the end of $t$, which also becomes a ($l$+1)-sequence \(t'\), represented as $<$$t\bigotimes i$$>$.
\end{definition}

As an extended version of the lexicographic sequence tree \cite{ayres2002sequential}, the LQS-Tree structure \cite{yin2012uspan,gan2020proum} is widely used in HUSPM to represent the search space. In this tree structure, the root is a null sequence, denoted as $<$$>$, whereas each of other nodes represents a candidate along with its utility. As its name implies, all children of a node are arranged in alphabetical order. The parent--child node relationship follows the regularity that children are either $I$-Extension or $S$-Extension sequences of their parent. 


As mentioned in Section 3.2, the utility of a sequence $t$ in a $q$-sequence $s$ is the maximum utility of all instances. However, the process of utility calculation may lead to long execution time for calculating multiple instances in one $q$-sequence. Take a sequence $t$: $<$\{\textit{a}\}, \{ \textit{b}\}$>$ and \(S_{2}\) in Table \ref{table1} as an example. First, we must find the positions of all instances of $t$ in \(S_{2}\): \(P(t,S_{2})\) = \{$<$1, 3$>$, $<$1, 4$>$, $<$2, 3$>$, $<$2, 4$>$\}. Then, we compute the utility of each instance and determine the maximum of the utilities of all instances, represented as \(u(t,S_{1})\) = max\{\$35, \$19, \$25, \$9\} = \$35. Thus, we adopt the utility chain data structure composed of utility lists to greatly reduce the calculation cost based on Ref. \cite{wang2016efficiently}. We discuss the details of utility chain in the following.

\begin{definition}[extension position]
	\rm Assume a sequence $t$ has an instance at position $<$$k_{1},k_{2},...,k_{m}$$>$ in a $q$-sequence $s$ with the extension position \(k_{m}\); the last item in $t$ is called the extension item. Note that different instances in one $q$-sequence can have the same extension position.
\end{definition}

Assume the set of extension positions of $t$ in $s$ is \{$p_{1},p_{2},\ldots,\ p_{n}$\}; then, the utility of $t$ at extension position \(p_{i}\) in the current $q$-sequence, denoted as $u(t,\ p_{i},\ s)$, is defined as $u(t,\ p_{i},\ s)$ = max\{$u$(t, $<$$ j_{1},j_{2},\ldots,p_{i}$$>$, s)$|j_{1} \leq j_{2} \leq \ldots \leq p_{i}$ and $<$$ j_{1},j_{2},\ldots,p_{i}$$> \in P(t,\ p_{i},\ s)$\}, where $P(t,\ p_{i},\ s)$ is the set of positions with extension position \(p_{i}\). 

According to the concept of extension position, we define the remaining utility \cite{yin2012uspan} of $t$ at extension position \(p_{i}\) in $s$ as $u_{rest}(t,\ p_{i},\ s)$ = $\sum_{i' \in s\bigwedge i' \succ t}^{}{u(i')}$, where \(i' \succ t\) indicates that the position of item \(i'\) is after the extension position $p_{i}$.

Consider the example in Table \ref{table1}, where $<$\{\textit{a}\}, \{\textit{b}\}$>$ has instances at positions $<$2, 4$>$ and $<$1, 4$>$ in \(S_{2}\); the extension positions of the two instances are both 4, and b is the extension item. Then, we can calculate $u$($<$\{\textit{a}\}, \{\textit{b}\}$>$, 4, $S_{2}$) = $\max$\{\$19, \$9\} = \$19 and $u_{rest}$($<$\{\textit{a}\}, \{\textit{b}\}$>$, 4, $S_{2}$) = \$4 + \$1 + \$3 + \$4 + \$5 = \$17.

\begin{definition}[pivot]
	\rm Assuming that the set of extension positions of a sequence $t$ in a $q$-sequence $s$ is \(\{ p_{1},p_{2},\ldots,\ p_{n}\}\), \(p_{1}\) is called the pivot of $t$ in $s$ \cite{yin2012uspan}.
\end{definition}

Based on the pivot concept, the remaining utility of $t$ in $s$ is denoted as \(u_{\text{rest}}(t,s)\) and defined as $u_{\text{rest}}(t,s)$ = $u_{rest}( t,\ p_{1},\ s)$. The pivot utility of $t$ in $s$, denoted as \(u_{p}(t,s)\), is defined as $u_{p}(t,s)$ = $u(t,\ p_{1},\ s)$.

For example, consider the sequence $<$\{\textit{a}\}, \{\textit{b}\}$>$ and \(S_{2}\) in Table \ref{table1}. We can calculate the remaining utility and  pivot utility as $u_{rest}$($<$\{\textit{a}\}, \{\textit{b}\}$>$, $S_{2}$) = $u_{rest}$($<$\{\textit{a}\}, \{\textit{b}\}$>$, 3, $S_{2}$) = \$17 and $u_{p}$($<$\{\textit{a}\}, \{\textit{b}\}$>$, $S_{2}$) = $u$($<$\{\textit{a}\}, \{\textit{b}\}$>$, 3, $S_{2}$) = \$35, where 3 is the pivot.

According to the concepts defined in this subsection, we designed the projected database to represent the necessary information of each sequence. As shown by the example in Figure \ref{prodatabase}, a projected database is composed of an information table and a utility chain consisting of a head table and multiple utility lists \cite{wang2016efficiently}. The information table includes two key pieces of information: the sequence itself and the prefix extension utility (\textit{PEU}) \cite{wang2016efficiently} value of the sequence in the projected database. The details are described in Section 4.3, but here we focus on the utility chain.

To facilitate the introduction of the projected database concept, we suppose a sequence $t$ has multiple instances at $m$ extension positions $PV(t,s)$ = $\{{pv}_{1},{pv}_{2},\cdots,{pv}_{m}\}$ in the $q$-sequence $s$, and we assume there are $n$ $q$-sequences, including $s$, containing $t$ in the database $D$. Each utility list is a set of elements with size of $m$, which correspond to $m$ extension positions of $t$ in $s$. The $i$th element contains the three following fields: 1) \textit{ItemsetID}, which is the $i$th extension position \({pv}_{i}\), 2) \textit{Utility}, which is the utility value of $t$ at the $i$th extension position \({pv}_{i}\), and 3) \textit{RestUtility}, which is the remaining utility of $t$ at the $i$th extension position \({pv}_{i}\). The head table stores two fields, \textit{SID} and \textit{PEU}, which record the identifier of $s$ and \textit{PEU} value of $t$ in $s$, respectively. There are $m$ rows in the head table, each of which corresponds to a utility list.

\begin{figure}[ht]
	\centering
	\includegraphics[clip,scale=0.56]{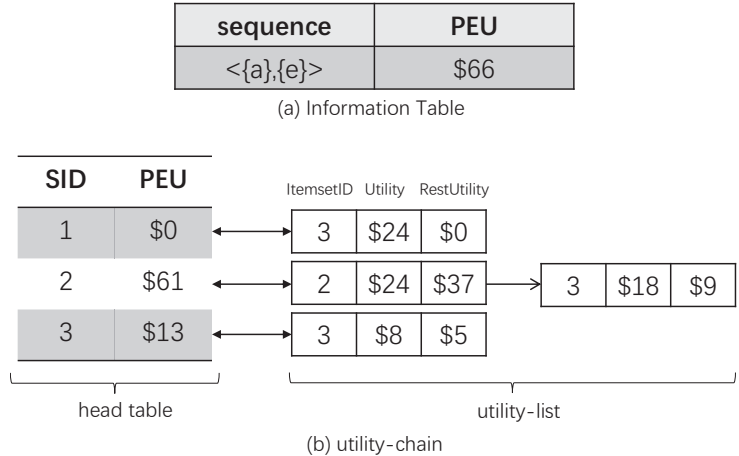}
	\caption{Example of the projected database of sequence \textless\{a\},\{e\}\textgreater{}}
	\label{prodatabase}
\end{figure}

The projected database of the sequence $<$\{\textit{a}\}, \{\textit{e}\}$>$ for the example in Table \ref{table1} is shown in Figure \ref{prodatabase}. The sequence $t$: $<$\{\textit{a}\}, \{\textit{e}\}$>$ is contained in \(S_{1}\), \(S_{2}\), and \(S_{3}\); hence, the head table has three rows that correspond to three utility lists. In \(S_{2}\), $t$ has multiple instances at extension positions 2 and 3, \(\ u_{p}(t,\ 2,\ S_{2})\) = \$24, \(u_{\text{rest}}(t,\ 2,\ S_{2})\) = \$37, \(u_{p}(t,\ 3,\ S_{2})\) = \$18, and \(u_{\text{rest}}(t,\ 3,\ S_{2})\) = \$9. Therefore, the second utility list is formed by two elements: (2, \$24, \$37) and (3, \$18, \$9).

Clearly, information included in the projected database of the sequence is precise and complete because the projected database of the $(l+1)$-sequence is built based on the parameter-free local search operation of $l$-sequences instead of a total scan of the whole original database. The TKUS algorithm adopts the projection and local search mechanism, which limits the scan space recursively by constructing a projected database of the current candidate and extending it. This divide-and-conquer concept greatly reduces the high costs of scanning, especially for a large database.

\subsection{Sequence Utility Raising Strategy}

As discussed in Section 1, to guarantee all top-$k$ patterns are found, TKUS employs a variable \textit{minutil}, which is increased from a very low value (close or equal to zero), as the threshold because the minimum utility threshold is not specified ahead of the mining process, which leads to more candidates being generated and the performance of the algorithm being degraded. To handle this problem, we propose the novel Sequence Utility Raising (SUR) strategy to raise the minimum utility threshold as fast as possible.

\begin{definition}[SUR strategy]
	\rm Given a $q$-sequence database $D$, let \textit{TKList} = \{$s_{1},s_{2},\ldots,s_{n}$\} be the set of all 1-, 2-, and $q$-sequences contained in $D$, where \(s_{i}\) (\(1 \leq i \leq n\)) is the sequence with the $i$th highest utility value in \textit{TKList}, denoted as \(u_{i}\). We can safely increase the variable \textit{minutil} to \(u_{k}\), where $k$ is the desired input number of patterns, before the mining task.
\end{definition}

Now, we prove that TKUS will not miss any top-$k$ HUSPs with the SUR strategy.

\begin{proof}
  \rm  Let \(H_{\textit{minutil}}\) be the set of HUSPs when the minimum utility threshold is set to \(\textit{minutil}\). It is obvious that \(H_{\textit{minutil}_{1}} \subseteq H_{\textit{minutil}_{2}}\), where \(\textit{minutil}_{1} > \textit{minutil}_{2}\). Suppose we increase the variable \textit{minutil} to \(u_{k}\) before the mining task; the set of sequences contained in \(H_{0}\) but not contained in \(H_{u_{k}}\), denoted as $R$, will not be discovered as unpromising candidates because their utility values are all lower than \(u_{k}\). For each sequence $s$ in $R$, there must be at least $k$ sequences (e.g., $k$ sequences from \(s_{1}\) to \(s_{k}\) in \textit{TKList}) whose utility values are not lower than the utility of $s$ in database $D$. Therefore, according to Definition 7, no sequences in $R$ are top-$k$ HUSPs and neglecting them does not affect the correctness of the results. Thus, TKUS will not miss any top-$k$ HUSPs with the SUR strategy.
\end{proof}

For example, in Table \ref{table1}, given $k$ = 4, the utility values of all 1-, 2-, and $q$-sequences can be calculated easily in one database scan: $u$($<$\{\textit{a}\}$>$) = \$50, $u$( $<$\{\textit{b}\}$>$) = \$36, $u$ $<$\{\textit{c}\}$>$) = \$10, $u$($<$\{\textit{c}\}, \{\textit{d}\}$>$) = \$9, $u$($<$\{\textit{a} \textit{d}\}$>$) = \$39, $u$($<$\{\textit{f}\}, \{\textit{a} \textit{d}\} \{\textit{d}\} \{\textit{a}\}$>$ = \$30, etc. As a result, we get the four sequences with the highest utility values $u$($<$\{\textit{a}\}, \{ \textit{a}\}$>$ = \$75, $u$($<$\{\textit{a} \textit{d}\}, \{\textit{a} \textit{e}\}, \{\textit{b} \textit{c} \textit{d} \textit{e}\}, \{\textit{b} \textit{d}\}$>$ = \$68, $u$($<$\{\textit{a}\}, \{\textit{e}\}$>$) = \$56, and $u$($<$\{\textit{a}\}$>$) = \$50, and \textit{minutil} increases to 50. It is clear that the SUR strategy generates the least unpromising candidates possible because it effectively increases the current minimum utility threshold to a rational level in advance.

\subsection{Upper Bounds and Corresponding Pruning Strategies}

Due to the monotonicity in the frequency-oriented framework \cite{gan2020fast}, it is computationally infeasible to reduce the search space by taking advantage of existing SPM upper bounds. In addition, compared with top-$k$ HUIM \cite{wu2012mining}, top-$k$ HUSPM intrinsically deals with a critical combinatorial explosion of the search space and computational complexity better. This is because inherent time order embedding in items leads to various possibilities of concatenation in the $q$-sequence database. To address this problem, several upper bounds have been proposed: sequence-weighted utilization (\textit{SWU}), sequence-projected utilization (\textit{SPU}), and sequence extension utility (\textit{SEU}); details are presented as follows. The basic principle of these upper bounds is to greatly reduce the search space by following the downward closure property, which is able to accelerate the mining process.

Suppose $t$ is a sequence and $D$ is a $q$-sequence database. Here, we briefly introduce the three upper bounds. \textit{SWU}, \textit{SPU}, and \textit{SEU} of $t$ in $D$ are defined as
\textit{SWU} = \(\sum_{t \subseteq s \land s \in D}^{}{u(s)}\) \cite{yin2012uspan}, \textit{SPU} = \(\sum_{t \subseteq s \land s \in D}^{}{(u_{p}\left( t,s \right) + u_{\text{rest}}(t,s)})\) \cite{yin2012uspan}, and  \textit{SEU} = \(\sum_{t \subseteq s \land s \in D}^{}{(u\left( t,s \right) +}u_{\text{rest}}(t,s))\) \cite{gan2020proum}, respectively. However, \textit{SPU} in not a true upper bound, which leads some real HUSPs being eliminated when pruning with it. The proof of this can be found in Ref. \cite{gan2020proum}. Therefore, the first top-$k$ HUSPM algorithm TUS \cite{yin2013efficiently}, which adopts the \textit{SPU} upper bound, may miss some top-$k$ HUSPs and return the wrong experimental results. More details of the downward closure property of an upper bound can also be found in Ref. \cite{gan2020proum}. \textit{SWU} and \textit{SEU} overestimate the utility of candidates and can limit the search to a reasonable scope. To further speed up the mining process, we adopt the tighter upper bounds of \textit{PEU} and \textit{RSU}, which were proposed in Ref. \cite{wang2016efficiently}.

\begin{definition}[\textit{PEU} upper bound \cite{wang2016efficiently}]
	\rm The \textit{PEU} of a sequence $t$ at position \(p: < k_{1},k_{2},...,k_{m} >\) in a $q$-sequence $s$, denoted as \(\textit{PEU}(t,p,s)\), is defined as 
	\[\textit{PEU}(t,\ p,\ s) = \left\{ \begin{matrix} u(t,\ p,\ s) + u_{\text{rest}}(t,\ k_{m},\ s),\ \ \ u_{\text{rest}}(t,\ k_{m},\ s) > 0 \\ 0\ \ \ \ \ \ \ \ \ \ \ \ \ \ \ \ \ \ \ \ \ \ \ \ ,\ \ \ otherwise \\
	\end{matrix} \right.\ \] 
	
	Based on \(\textit{PEU}(t,p,s)\), the \textit{PEU} of $t$ in the $q$-sequence is defined as \(\textit{PEU}(t,s)\) = \(\max\{\textit{PEU}\left( t,\ p_{i},\ s \right)\}\), where \(p_{i}\) is one position of $t$ in $s$ \(P(t,s)\). Moreover, the overall \textit{PEU} value of the sequence $t$ in a $q$-sequence database $D$ is the sum of the \textit{PEU}s of $t$ in each $q$-sequence, which is denoted as \(\textit{PEU}(t)\) and defined as \(\textit{PEU}(t)\) = \(\sum_{t \sqsubseteq s \land s \in D}^{}{\textit{PEU}(t,s)}\).
\end{definition}

For example, consider the sequence $t$: $<$\{\textit{a}\}, \{\textit{e}\}$>$ in Figure \ref{prodatabase}, where $t$ has two instances in \(S_{1}\). \textit{PEU}($t$, $<$1, 3$>$, $S_{1})$ and \textit{PEU}($t$, $<$2, 3$>$, $S_{1})$ are both equal to \$0 because \(u_{\text{rest}}( t,\ 3,\ S_{1})\) is \$0. In addition, $t$ has three instances in \(S_{2}\), and we have \(\textit{PEU}\left( t,S_{2} \right)\) = $\max$\{\$61, \$27, \$17\} = \$61. Finally, the \textit{PEU} of any sequence extended by $t$ can be calculated as follows: \(\textit{PEU}(t)\) = \(\textit{PEU}(t,S_{1}) + \textit{PEU}\left( t,S_{2} \right) + \textit{PEU}(t,S_{3})\) = \$0 + \$61 + \$13 = \$74.

\begin{theorem}
	Given two sequences $t$ and \(t'\) and a $q$-sequence database $D$, suppose \(t'\) is a prefix of $t$. We have
	\(u(t) \leq \textit{PEU}(t')\) \cite{wang2016efficiently}.
\end{theorem}

\begin{proof}
   \rm	Because \(t'\) is a prefix of $t$, let \(t\) = \(t' \bullet t''\), where \({|t}''| > 0\). Assume $t$ is contained in a $q$-sequence $s$. The utility of $t$ in $s$ can be divided into two parts denoted as \(u(t,s)\) = \(u( t',\ p,\ s) + u_{\textit{limit}}(t'')\), where $p$ is one of the extension positions of \(t'\) in $s$, and \(u_{\textit{limit}}(t'')\) is the utility of an instance of \(t''\) in $s$ with the first item after position $p$. Clearly, $u_{\textit{limit}}( t'') \leq u_{\text{rest}}(t',\ p,\ s)$; then, we have
	\begin{align*}
	u(t,s) &= u(t',\ p,\ s) + u_{\textit{limit}}(t'')  \\
	&\leq u(t',\ p,\ s) + u_{\text{rest}}(t',\ p,\ s)  \\
	&\leq max\{u(t',\ p,\ s) + u_{\text{rest}}(t',\ p,\ s)\}  \\
	&\leq \textit{PEU}(t',s)
	\end{align*}
	Because \(t' \sqsubseteq t\), we have
	\(u(t)\) = \(\sum_{s \in D \land t \sqsubseteq s}^{}u(t,s)\) $\leq$ \(\sum_{s \in D \land t \sqsubseteq s}^{}{}\textit{PEU}\left( t',s \right)\) $\leq$ \(\sum_{s \in D \land t' \sqsubseteq s}^{}{}\textit{PEU}( t',s)\) = \(\textit{PEU}(t')\)
	in a database $D$.
\end{proof}

Based on the aforementioned proof, we can draw the conclusion that the utilities of all sequences that can be extended from $t$ are lower than a given value if the \textit{PEU} of $t$ is lower than that value. Thus, we have the following TDE pruning strategy.

\begin{definition}[TDE pruning strategy]
	\rm In the TKUS algorithm, let $t$ and \textit{minutil} be a candidate and the current minimum utility threshold, respectively. Assume that the node $N$ in the LQS-tree denotes the candidate $t$. If \(\textit{PEU}(t) \leq \textit{minutil}\), then TKUS need not check its descendants; in other words, TKUS can terminate the present iteration from node $N$.
\end{definition} 

Clearly, the TDE pruning strategy is a depth-based strategy and can prevent scanning deep but unpromising nodes in the LQS-tree. Significantly, it is a safe strategy because the \textit{PEU} upper bound has the downward closure property; hence, any descendant of the node denoting $t$ cannot be a desired top-$k$ HUSP. Consider the sequence $t$: $<$\{\textit{a}\}, \{\textit{e}\}$>$ for the example in Table \ref{table1}. We have calculated the \textit{PEU} value of $t$ in Figure \ref{prodatabase}, where \(\textit{PEU}(t)\) = \$74. Suppose the minimum utility threshold has already been increased to 80 (i.e., \textit{minutil} = 80); then, the algorithm terminates at the node representing $t$ and stops searching its descendants according to the TDE pruning strategy.

\begin{definition}[\textit{RSU} upper bound \cite{wang2016efficiently}]
	\rm Suppose \(\alpha\) is a sequence that can be extended to the sequence $t$ through one $I$-Extension or $S$-Extension operation. In other words, the node representing \(\alpha\) is the parent of the node representing $t$ in the LQS-tree. The \textit{RSU} of $t$ in a $q$-sequence $s$, denoted as \(RSU(t,s)\), is defined as
	
	\[\textit{RSU}(t,s) = \left\{ \begin{matrix} \textit{PEU}(\alpha,s),\ \ & t \sqsubseteq s \land \alpha \sqsubseteq s \\ 0\ \ \ \ \ \ \ ,\ \ & otherwise \\
	\end{matrix} \right. \ \]
	
	Then, the \textit{RSU} of $t$ in database $D$ can be denoted as \(\textit{RSU}(t)\) and defined as \(\textit{RSU}( t)\) = \(\sum_{t \sqsubseteq s \land s \sqsubseteq D}^{}{\textit{RSU}(t,s)}\).
\end{definition}

Consider the example in Table \ref{table1}. Suppose $t$: $<$\{\textit{a}\}, \{\textit{e}\}, \{\textit{f}\}$>$ and $\alpha$: $<$\{\textit{a}\}, \{\textit{e}\}$>$. It is clear that only \(S_{2}\) contains both $t$ and \(t'\). According to Definition 18, we have \textit{RSU}$(t, S_{1})$ = \textit{RSU}$(t, S_{3})$ = \$0, because \(S_{1}\) and \(S_{3}\) contain \(\alpha\) but do not contain $t$, and \textit{RSU}$(t, S_{2})$ = \textit{PEU}$(\alpha, S_{2})$ = \$61. Finally, the \textit{RSU} of $t$ is \textit{RSU}$(t)$ = \$61.

\begin{theorem}
	Given two sequences $t$ and \(t'\) and a $q$-sequence database $D$, assume \(t'\) is a prefix of $t$ or \(t' = t\). We have \(u(t) \leq \textit{RSU}(t')\) \cite{wang2016efficiently}.
\end{theorem}

\begin{proof}
   \rm	Suppose \(\alpha\) is a sequence that can be extended to \(t'\) through one $I$-Extension or $S$-Extension operation. \(\alpha\) is also a prefix of $t$ because \(t'\) is a prefix of $t$ or \(t'\) = \(t\). Given a $q$-sequence $s$ containing $t$, we obtain \(u(t, s) \leq \textit{PEU}(\alpha, s)\) according to the proof of Theorem 1. Based on Definition 13, we have \(\textit{RSU}(t', s) = \textit{PEU}(\alpha,s)\). Finally, we have \(u(t, s) \leq \textit{RSU}(t', s)\), so we can conclude that \(u(t) \leq \textit{RSU}(t')\) in database $D$.
\end{proof}

Different from the \textit{PEU} upper bound, \textit{RSU} is an overestimation of both its descendant and itself. Now, we have the following EUI pruning strategy.

\begin{definition}[EUI pruning strategy]
	\rm In the TKUS procedure, let $t$ and \textit{minutil} be a candidate and the current minimum utility threshold, respectively. Assume the set \(I\): \(\{ i_{1}, i_{2}, \ldots, i_{n}\}\) consists of all items that can be extended to $t$ to form \(t'\). Suppose $i_{m}$ \((1 \leq m \leq n)\) is extended to $t$ forming \(t'\); if \(\textit{RSU}( t') \leq \textit{minutil}\), then TKUS can be stopped from exploring node $N$. 
\end{definition}

It should be noted that EUI is a width-based pruning strategy, which is able to stop TKUS from scanning a new branch in the LQS-tree. The EUI strategy is also a safe strategy, ensuring that all top-$k$ HUSPs are obtained. This is because the \textit{RSU} upper bound is monotonous, so any descendants of $t$ and itself cannot be contained in the set of top-$k$ HUSPs. Consider the sequence $t$: $<$\{\textit{a}\}, \{\textit{e}\}$>$ and one of its $S$-Extension items \textit{f} in Table \ref{table1} as an instance. We calculated the \textit{RSU} value of $t'$ = $<$$t \otimes f$$>$ = $<$\{\textit{a}\}, \{\textit{e}\}, \{\textit{f}\}$>$ in Figure \ref{prodatabase}, where \(\textit{RSU}(t')\) = \$61. Suppose the minimum utility threshold has been increased to \$65 (i.e., \(\textit{minutil}\) = \$65); the algorithm prunes the node representing \(t'\) without calculating its utility value, as well as its descendants according to the EUI pruning strategy.

\subsection{Proposed TKUS algorithm}

Based on the data structure and some aforementioned strategies, the proposed TKUS algorithm is described as follows. The algorithmic form of the TKUS algorithm is given in Algorithm \ref{alg:The TKUS Algorithm} and Algorithm \ref{alg:Projection-TKUS}, which represent the main and recursive procedures, respectively.

\begin{algorithm}[ht]
	\caption{The TKUS Algorithm}
	\label{alg:The TKUS Algorithm}
	\begin{algorithmic}[1]
		\REQUIRE 
		$D$: a quantitive database; 
		$k$: the number of desired HUSPs; 
		$utable$: a table containing the external utility of each item;
		\ENSURE 
		Top-$k$ HUSPs in $D$;
		
		\STATE initialize \textit{minutil} = 0;
		\STATE first scan the original database $D$ to: 
		1). calculate the utility values of all 1-squences and 2-sequences; 
		2). store the utility values of all $q$-sequences;
		3). construct projected databases of all 1-sequences; \qquad(\textbf{The projection and local search mechanism})
		
		\STATE sort the utility values of 1-sequences, 2-sequences and $q$-sequences in descending order, and get the $k$-th highest utility value $\textit{minutil}_{0}$; 
		\STATE set \textit{minutil} = $\textit{minutil}_{0}$; \qquad(\textbf{The SUR strategy}) 
		\STATE initialize \textit{TKList} = $\emptyset$; 
		\FOR {$t \in $ 1-sequences }
			\IF{$t.\textit{utility} \ge \textit{minutil}$}
			\STATE update \textit{TKList};
			\ENDIF
		\ENDFOR
		\FOR {$t \in $ 1-sequences }
			\IF{$\textit{PEU}(t) \ge \textit{minutil}$}
			\STATE call \textit{Projection-TKUS(t,$\text{DB}_s$)}; \qquad(\textbf{The TDE strategy})
			\ENDIF
		\ENDFOR
		
		\STATE \textbf{return} \textit{TKList};
		
	\end{algorithmic}	
\end{algorithm}

\begin{algorithm}[ht]
	\caption{Projection-TKUS}
	\label{alg:Projection-TKUS}
	\begin{algorithmic}[1]
		\REQUIRE 
		$t$: a sequence as prefix; 
		$\textit{DB}_{t}$: the projected database of $t$; 
		\STATE scan $\textit{DB}_{t}$ once to:
		1). add $I$-Extension items of $t$ to \textit{ilist}; 2).add $S$-Extension items of $t$ to \textit{slist};
		
		\FOR {each item $i \in $ \textit{ilist} }
			\STATE $t'$ $\leftarrow$ $<$$t\bigoplus i$$>$; 
			\IF{$\textit{RSU}(t')<\textit{minutil}$}
				\STATE remove $i$ from \textit{ilist}; \qquad(\textbf{The EUI strategy})
			\ENDIF
			\STATE construct projection database of $t'$ $\textit{DB}_{t'}$; \qquad(\textbf{The projection and local search mechanism})
			\IF{$t'.utility\ge \textit{minutil}$}
				\STATE update \textit{TKList};
				\STATE update \textit{minutil};
			\ENDIF
			\STATE put $t'$ to \textit{seqlist};
		\ENDFOR
		\FOR {each item $i \in $ \textit{slist} }
			\STATE $t'$ $\leftarrow$ $<$$t\bigotimes i$$>$;
			\IF{$RSU(t')<\textit{minutil}$}
				\STATE remove $i$ from \textit{slist} \qquad(\textbf{The EUI strategy})
			\ENDIF
			\STATE construct projection database of $t'$ $\textit{DB}_{t'}$; \qquad(\textbf{The projection and local search mechanism})
			\IF{$t'.utility\ge \textit{minutil}$}
				\STATE update \textit{TKList};
				\STATE update \textit{minutil};
			\ENDIF
			\STATE put $t'$ to \textit{seqlist};
		\ENDFOR
		\STATE sort all sequences in \textit{seqlist} according to \textit{PEU} in descending order;
		\FOR {each sequence $s \in $ \textit{seqlist} }
			\IF{$PEU(t)\ge \textit{minutil}$}
				\STATE call \textit{Projection-TKUS(s,s.DB)}; \qquad(\textbf{The TDE strategy})
			\ENDIF
		\ENDFOR
	\end{algorithmic}	
\end{algorithm}

In the beginning of the mining process, the TKUS algorithm adopts the variable \textit{minutil} to store the current minimum utility threshold (Line 1). It first scans the $q$-sequence database $D$ to calculate the utility values of all 1-sequences (i.e., items) and 2-sequences, store the utility values of all $q$-sequences, and construct the projected databases of all 1-sequences (Line 2). Then, following the SUR strategy, it increases \textit{minutil} to the $k$th highest utility value from the utility values obtained before (Lines 3--4). In addition, it engages a sorted list called \textit{TKList} with a fixed size of $k$ to maintain the top-$k$ HUSPs on-the-fly (Line 5). In addition, for each 1-sequence, if its utility value is greater than \textit{minutil}, the \textit{TKList} data structure is updated (Lines 6--8). After that, adopting the TDE strategy, TKUS recursively builds projected databases and searches locally by traversing the LQS-tree with the depth-first search strategy (Lines 9--11).

As shown in Algorithm \ref{alg:Projection-TKUS}, the \textit{Projection-TKUS} procedure follows a width-first search to enumerate sequences, which can be generated by the prefix according to the projection and local search mechanism. The algorithm regards each input sequence $t$ as a prefix and scans the projected database to find items that can be extended (Line 1). It is worth noting that the \textit{RSU} value of each extension sequence \(t'\) is calculated simultaneously during scanning. Items in \textit{ilist} will be processed in the following way (Lines 2--10). For the extended sequence \(t'\), TKUS checks whether it could be a top-$k$ HUSP or a prefix of a top-$k$ HUSP using the EUI strategy (Lines 4--5). Then, TKUS builds the projected database of \(t'\) to reduce the search space (Line 6). Meanwhile, the utility value of \(t'\) is calculated when searching the space locally. If the utility value of \(t'\) is greater than \textit{minutil}, TKUS updates the \textit{TKList} and \textit{minutil} because \(t'\) may be a top-$k$ HUSP (Lines 7--10). Similarly, TKUS addresses each item in \textit{slist}  analogous to the aforementioned procedure (Lines 11--19). After the extension operation, TKUS sorts \textit{seqlist}, the list of all generated sequences whose \textit{RSU} values are larger than \textit{minutil}, in descending order according to \textit{PEU} (Line 20). This is because the utility values of a sequence's descendants are intuitively more likely to be higher than those of lower-utility ones. Finally, if the \textit{PEU} of the sequence in \textit{seqlist} is greater than the minimum utility threshold \textit{minutil}, then the \textit{Projection-TKUS} is applied recursively to discover the top-$k$ HUSPs by considering the descendants of the sequence (Lines 21--23).

\section{Experiments} 
\label{sec:experiments}

To assess the performance of the proposed TKUS algorithm, we conducted experiments implemented in Java and developed in IntelliJ IDEA2019. All experiments were performed on a personal computer equipped with a 3.8 GHz Intel Core i7-10700K CPU, 32 GB RAM, and 64-bit Windows 10 operating system. Specifically, the TKUS algorithm was evaluated against the state-of-the-art TKHUS-Span algorithm \cite{wang2016efficiently}. The only other existing top-$k$ HUSPM algorithm, TUS \cite{yin2013efficiently}, was not evaluated because it is incapable of ensuring the discovery of all top-$k$ HUSPs, as discussed in Section 4.3. Details of the experiments are given in the following.

\subsection{Datasets and Data Preprocessing}

To evaluate the performance of the algorithms, we first verified TKUS on six datasets, including five real-world datasets. We used two publicly available real-world e-commerce datasets, which are generated from a series of purchase records, in our experiments. In addition, we utilized three linguistic datasets, which can be easily transformed by an article or a book. These datasets represent the main categories of data with varied features that are typically encountered in real-world scenarios.

Moreover, we used a synthetic dataset called Syn80K, where each $q$-itemset contains more $q$-items on average. The features of the used datasets are listed in Table \ref{features1} and Table \ref{feature2}. Note that the Yoochoose dataset is available at the RecSys website\footnote{\url{https://recsys.acm.org/recsys15/challenge/}}, whereas the other four datasets can be obtained from the open-source data mining library\footnote{\url{http://www.philippe-fournier-viger.com/spmf/}}. More details about datasets we used can be found in Ref. \cite{gan2020proum}. 

\begin{table*}[ht]
	\caption{Properties of datasets}
	\label{features1}
	\centering
		\begin{tabular}{|c|c|}
			\hline
			\textbf{Feature} & \textbf{Description} \\ \hline \hline      
			$|D|$ & number of $q$-sequences  \\ \hline
			$|I|$ & number of different items  \\ \hline 
			AVG($S$) & average length of $q$-sequence \\ \hline
			MAX($S$) & maximum length of $q$-sequence  \\ \hline
			AVG($Sequence$)	& average number of $q$-itemsets per $q$-sequence  \\ \hline
			AVG($Itemset$)	&  average number of $q$-items per $q$-itemset  \\ \hline
		\end{tabular}		
\end{table*}

\begin{table*}[ht]
	\caption{Details of features of datasets}
		\label{feature2}      
		\centering
		\begin{tabular}{|c|c|c|c|c|c|c|}
			\hline
			\textbf{Dataset} & \textbf{$|D|$} & \textbf{$|I|$} & \textbf{AVG(\textit{S})} & \textbf{MAX(\textit{S})} & \textbf{AVG(\textit{Sequence})} & \textbf{AVG(\textit{Itemset})} \\ \hline   \hline    
			Bible & 36369 & 13905	& 21.64 & 100 & 17.85 & 1.0 \\ \hline
			Leviathan & 5834 & 9025 & 33.81 & 100 & 26.34 & 1.0 \\ \hline
			Sign & 730 & 267	& 52 & 94 & 51.99 & 1.0 \\ \hline
			Yoochoose & 234300 & 16004 & 1.13 & 21 & 2.11 & 1.97 \\ \hline
			Kosarak	& 10000 & 10094	& 8.14 & 608 & 8.14 & 1.0 \\ \hline
			Syn80K	& 79718 & 7584	& 6.19 & 18 & 26.69 & 4.32 \\ \hline
		\end{tabular}			
\end{table*}

$ \bullet $ \textit{\textbf{Bible}} is a conversion of the Bible into a dataset, where each word can be seen as an item, while each sentence is regarded as a sequence.

$ \bullet $ \textit{\textbf{Leviathan}} is a conversion of the novel Leviathan by Thomas Hobbes. The generation method is similar to that of the previous dataset, where we adopt a digital item to replace a special word.

$ \bullet $ \textit{\textbf{Sign}} is a sign language utterance dataset created by the National Center for Sign Language and Gesture Resources at Boston University. Each utterance has a connection with a segment of video.

$ \bullet $ \textit{\textbf{Yoochoose}} is a very large dataset from an e-commerce site. Each click-stream of e-commodities is transformed to a transaction, which sometimes ends in a purchase event. Each record has five fields: Session ID, Timestamp, Item ID, Price, and Quantity. 

$ \bullet $ \textit{\textbf{Kosarak}} is also a click-stream dataset from a Hungarian online news portal. Some records in this dataset have extremely high lengths, which increases the mining difficulty.

\subsection{Effectiveness Analysis}

\begin{table*}[!htbp]
	\centering
	\caption{Effectiveness of different strategies utilizing by TKUS}
	\label{table_effect}
	\begin{tabular}{|c|c|p{2cm}<{\centering}|p{2cm}<{\centering}|p{2cm}<{\centering}|p{2cm}<{\centering}|}
		\hline \textbf{Dataset} &  \textbf{Result} & $\mathbf{TKUS}$ & $\mathbf{TKUS_{SUR}}$ & $\mathbf{TKUS_{TDE}}$ & $\mathbf{TKUS_{EUI}}$ \\
		\hline \hline
		\multirow{3}{*}{\makecell[c]{Bible \\ $k$ = 200}} 
		& {Time} & {888} & {936} & {/} & {2,028} \\
		\cline{2-2}
		& {Memory} & {4,247} & {3,634} & {/} & {2,748} \\
		\cline{2-2}
		& {Candidate} & {17,059} & {18,540} & {/} &{19,462} \\
		\hline
		
		\multirow{3}{*}{\makecell[c]{Leviathan \\ $k$ = 500}} 
		& {Time} & {339} & {346} & {1,090} & {542} \\
		\cline{2-2}
		& {Memory} & {2,893} & {3,089} & {2,963} & {3,071} \\
		\cline{2-2}
		& {Candidate} & {40,918} & {42,048} & {37,348,596} &{40,918} \\
		\hline
		
		\multirow{3}{*}{\makecell[c]{Sign \\ $k$ = 500}}
		& {Time} & {193} & {190} & {3,147} & {231} \\
		\cline{2-2}
		& {Memory} & {2,682} & {3,669} & {2,013} & {2,846} \\
		\cline{2-2}
		& {Candidate} & {193,477} & {194,895} & {164,174,202} &{193,477} \\
		\hline
		
		\multirow{3}{*}{\makecell[c]{Yoochoose \\ $k$ = 8000}}
		& {Time} & {35} & {22} & {23} & {27} \\
		\cline{2-2}
		& {Memory} & {2,780} & {3,579} & {2,146} & {2,351} \\
		\cline{2-2}
		& {Candidate} & {180,641} & {281,419} & {832,550} &{204,694} \\
		\hline
		
		\multirow{3}{*}{\makecell[c]{Kosarak \\ $k$ = 9}} 
		& {Time} & {13} & {15} & {/} & {23} \\
		\cline{2-2}
		& {Memory} & {2,279} & {3,540} & {/} & {2,293} \\
		\cline{2-2}
		& {Candidate} & {2,001} & {2,470} & {/} &{2,001} \\
		\hline
		
		\multirow{3}{*}{\makecell[c]{Syn80K \\ $k$ = 200}} 
		& {Time} & {839} & {936} & {/} & {3,267} \\
		\cline{2-2}
		& {Memory} & {5,483} & {5,982} & {/} & {4,184} \\
		\cline{2-2}
		& {Candidate} & {972,569} & {1,004,164} & {/} &{972,569} \\
		\hline
					
		\hline
	\end{tabular}
\end{table*}

\begin{figure*}[ht]
	\centering
	\includegraphics[clip,scale=0.5]{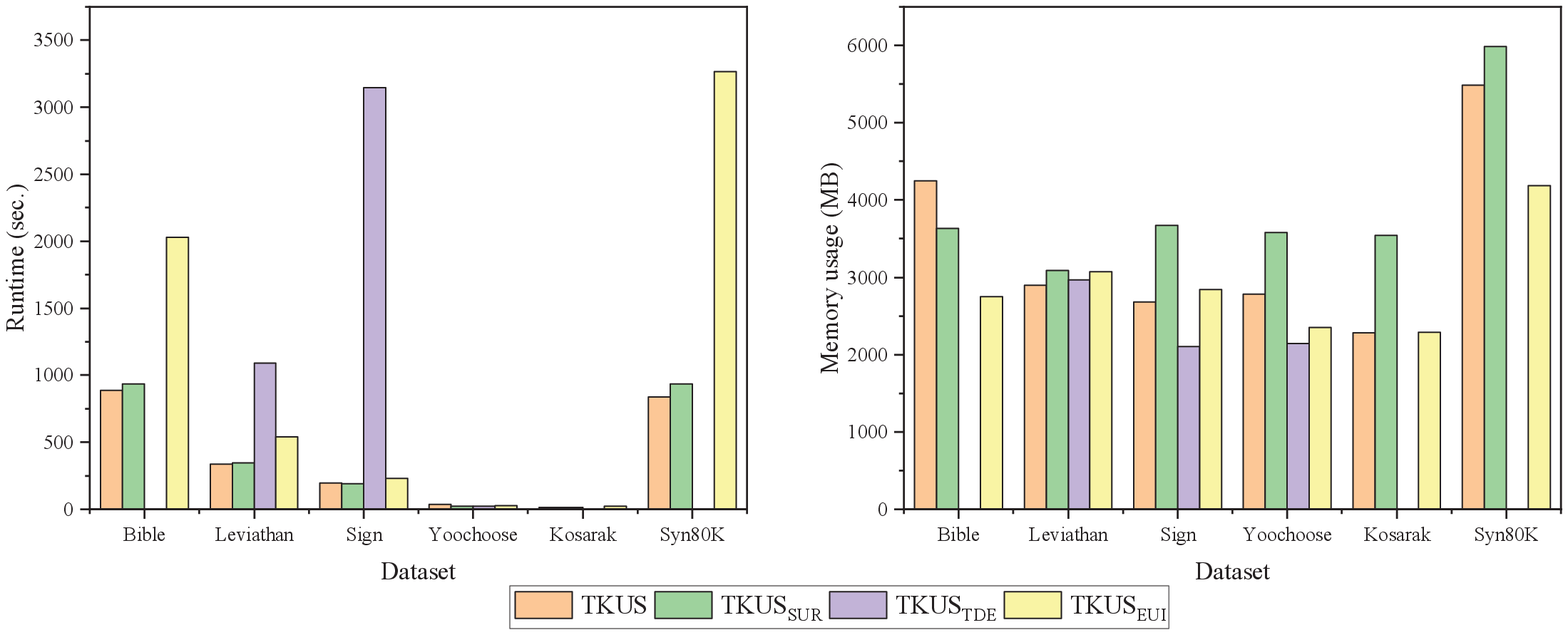}
	\caption{Effectiveness of different strategies utilizing by TKUS}
	\label{effect}
\end{figure*}

In this subsection, to analyze the effects of several designed mining strategies, the number of candidates and performance are evaluated in terms of running time and memory consumption of the algorithms. Moreover, for a fair comparison, we adopted TKUS as the backbone algorithm for effectiveness evaluation. We investigated the proposed variants to examine the impact of each of SUR, TDE, and EUI. The three variants $\rm TKUS_{SUR}$, $\rm TKUS_{TDE}$, $\rm TKUS_{EUI}$ are the TKUS algorithm without the corresponding pruning techniques. We conducted the experiments on six datasets with $k$ set to 200, 500, 500, 8000, 9, and 200, respective to the order they were introduced in. The results are listed in Table \ref{table_effect}. To better visualize the results, we also illustrate the performance of the three variants and TKUS itself in Figure \ref{effect}. It is clear that the variants of TKUS cannot achieve very good performance in most cases. On the Yoochoose dataset, the two variants have a slightly better performance because scanning many small projected databases to calculate the \textit{PEU} and \textit{RSU} values may be time-consuming for frequent disk I/O accesses. It is also clear that TKUS benefits most from the TDE strategy because $\rm TKUS_{TDE}$ generates the most candidates and could not even return the top-$k$ HUSPs in the limited running time (i.e., three hours) on the Bible, Kosarak, and Syn80K datasets. It is worth noting that we compared the results returned by TKUS and its three variants; as expected, the results are the same in the same scenarios, which proves that none of the mining strategies result in missing any top-$k$ HUSPs and ensures the completeness of the TKUS algorithm. In summary, it can be concluded that the three strategies all contribute to the efficiency of TKUS owing to the proposed method outperforming each individually.

\subsection{Candidate and Minimum Utility Threshold Analysis}

\begin{table}[ht]
	\centering
	\caption{Number of candidates generated by compared methods}
	\label{table_threshold}
	\begin{tabular}{|c|c|c|c|c|c|c|c|}
		\hline \textbf{Dataset} & \textbf{Result} & $\mathbf{test_{1}}$ & $\mathbf{test_{2}}$ & $\mathbf{test_{3}}$ & $\mathbf{test_{4}}$ & $\mathbf{test_{5}}$ & $\mathbf{test_{6}}$ \\
		\hline  \hline
		\multirow{4}{*}{Bible} 
		& {$k$} & {100} & {200} & {500} & {1000} & {2000} & {3000} \\
		\cline{2-2}
		& {TKHUS-Span} & {9,091} & {18,540} & {45,796} & {90,824} & {180,052} & {274,127} \\
		\cline{2-2}
		& {TKUS} & {8,405} & {17,059} & {40,887} & {80,567} &  {163,659} & {246,976} \\
		\cline{2-2}
		& {\textit{SSSR}} & {7.55\%} & {7.99\%} & {10.72\%} & {11.29\%} & {9.10\%} & {9.90\%} \\
		\hline
		
		\multirow{4}{*}{Leviathan}
		& {$k$} & {100} & {200} & {500} & {1000} & {2000} & {3000} \\
		\cline{2-2}
		& {TKHUS-Span} & {10,921} & {19,304} & {42,048} & {77,281} & {145,570} & {216,208} \\
		\cline{2-2}
		& {TKUS} & {10,859} & {18,972} & {40,918} & {74,655} &  {138,846} & {203,741} \\
		\cline{2-2}
		& {\textit{SSSR}} & {0.57\%} & {1.72\%} & {2.69\%} & {3.40\%} & {4.62\%} & {5.77\%} \\
		\hline

		\multirow{4}{*}{Sign} 
		& {$k$} & {100} & {200} & {500} & {1000} & {2000} & {3000} \\
		\cline{2-2}
		& {TKHUS-Span} & {90,988} & {125,252} & {194,895} & {280,759} & {458,508} & {693,456} \\
		\cline{2-2}
		& {TKUS} & {90,759} & {124,699} & {193,477} & {276,074} &  {407,596} & {524,255} \\
		\cline{2-2}
		& {\textit{SSSR}} & {0.25\%} & {0.44\%} & {0.73\%} & {1.67\%} & {11.10\%} & {24.40\%} \\
		\hline
		
		\multirow{4}{*}{Yoochoose} 
		& {$k$} & {2000} & {4000} & {6000} & {8000} & {10000} & {12000} \\
		\cline{2-2}
		& {TKHUS-Span} & {4,994} & {56,299} & {173,900} & {281,419} & {381,123} & {482,225} \\
		\cline{2-2}
		& {TKUS} & {4,329} & {47,856} & {97,820} & {180,641} &  {314,307} & {431,406} \\
		\cline{2-2}
		& {\textit{SSSR}} & {13.32\%} & {15.00\%} & {43.75\%} & {35.81\%} & {17.53\%} & {17.53\%} \\
		\hline
		
		\multirow{4}{*}{Kosarak} 
		& {$k$} & {9} & {10} & {11} & {12} & {13} & {14} \\
		\cline{2-2}
		& {TKHUS-Span} & {2,470} & {3,358} & {3,751} & {4,542} & {142,502} & {387,811} \\
		\cline{2-2}
		& {TKUS} & {2,001} & {2,821} & {3,123} & {3,757} &  {5,183} & {6,022} \\
		\cline{2-2}
		& {\textit{SSSR}} & {18.99\%} & {15.99\%} & {16.74\%} & {17.28\%} & {96.36\%} & {98.45\%} \\
		\hline
		
		\multirow{4}{*}{Syn80K} 
		& {$k$} & {100} & {200} & {500} & {1000} & {2000} & {3000} \\
		\cline{2-2}
		& {TKHUS-Span} & {355,442} & {1,004,164} & {3,133,862} & {5,836,250} & {9,584,865} & {12,783,180} \\
		\cline{2-2}
		& {TKUS} & {341,889} & {972,569} & {2,993,034} & {5,782,992} &  {9,583,074} & {12,782,996} \\
		\cline{2-2}
		& {\textit{SSSR}} & {3.81\%} & {3.15\%} & {4.49\%} & {0.91\%} & {0.02\%} & {0.01\%} \\
		\hline

		\hline
	\end{tabular}
\end{table}

\begin{figure}[ht]
	\centering
	\includegraphics[clip,scale=0.7]{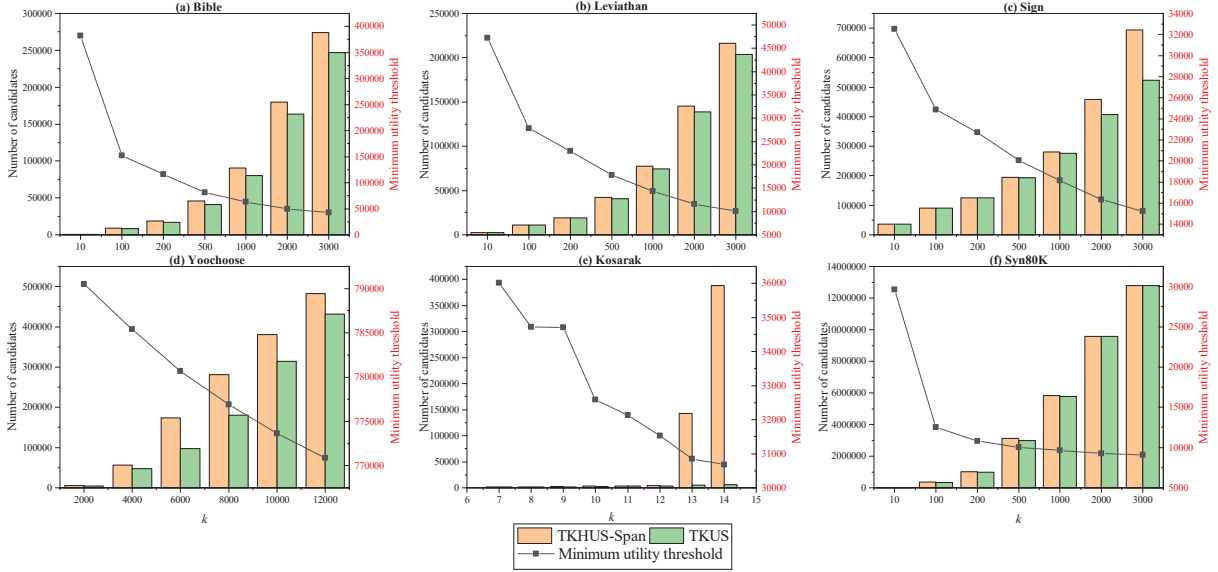}
	\caption{Number of candidates and final minimum utility threshold of top-$k$ HUSPs}
	\label{threshold}
\end{figure}

First, the numbers of candidates of the two algorithms are evaluated, which is a crucial measure of the search space. The results are listed in Table \ref{table_threshold}. To better visualize the results, we also illustrate related results in Figure \ref{threshold} with the final minimum utility thresholds. In particular, we also define a novel metric called search space shrinkage rate (\textit{SSSR}) to evaluate how much TKUS can reduce the search space compared to TKHUS-Span. Given a $q$-sequence dataset $D$ and desired number of HUSPs $k$, assume that $\textit{Num}_{\textit{TKHUS-Span}}$ and $\textit{Num}_{\textit{TKUS}}$ are the numbers of candidates generated in the mining processes of TKHUS-Span and TKUS, respectively. The \textit{SSSR} value is define as $\textit{SSSR}(D,k)$ = $\frac{\textit{Num}_{\textit{TKHUS-Span}} - \textit{Num}_{\textit{TKUS}}}{\textit{Num}_{\textit{TKHUS-Span}}}$. 

From a macro perspective, it is clear that the number of candidates mainly depends on the characters of the datasets. More specially, the more $q$-sequences or longer the average $q$-sequence, the more candidates will be generated by the algorithms. In addition, the quantity of candidates becomes larger with increasing $k$. Moreover, TKUS generates less candidates than TKHUS-Span in all cases because the projection and local search mechanism as well as the three mining strategies make a vital contribution to reducing the search space and improving the efficiency. In terms of the \textit{SSSR} metric, TKUS has better performance than TKHUS-Span, especially on the Bible, Leviathan, Yoochoose, and Kosarak datasets because the metric values are all larger than 5\%. Surprisingly, our designed TKUS can avoid scanning more than 98 percent of the search space of TKHUS-Span on the Kosarak dataset when setting $k$ to 14. Furthermore, as shown in Figure \ref{effect}, the final minimum utility threshold shows a decline following the increasing of $k$ in all datasets. In summary, we can conclude that the developed local search mechanism and strategies greatly limit the scope of scanning and reduce the number of candidates significantly.

\subsection{Efficiency Analysis}

The efficiency of the algorithm is measured by its execution time considering not only the running time used by the CPU but also the access time required for disk input/output. The execution times of the TKHUS-Span algorithm compared to the proposed TKUS under different $k$ on the six datasets are shown in Figure \ref{runtime}. From the experimental results, it can be seen that the TKUS algorithm performs better overall for all datasets. In particular, with increasing $k$ on the Kosarak dataset, the execution time of TKHUS-Span dramatically increases when $k$ becomes greater than 12, whereas that of TKUS remains relatively stable. This is expected as the proposed strategies help to improve the efficiency of TKUS. However, the two algorithms take negligible time to finish top-$k$ HUSPM with a comparatively small value of $k$. Moreover, TKUS has bad performance on the Yoochoose dataset, and we speculate that the massive calculations of 1- and 2-sequences ahead of the mining incurred very high computational costs. On the other four datasets, the execution time increases with increasing $k$, where the growth curve clearly resembles a logarithmic function. The difference is that TKHUS-Span requires more running time with the same value of $k$, and the gap grows with increasing $k$. This demonstrates that the developed projection and local search mechanism and several strategies play a key role in improving the performance.

\begin{figure}[ht]
	\centering
	\includegraphics[clip,scale=0.7]{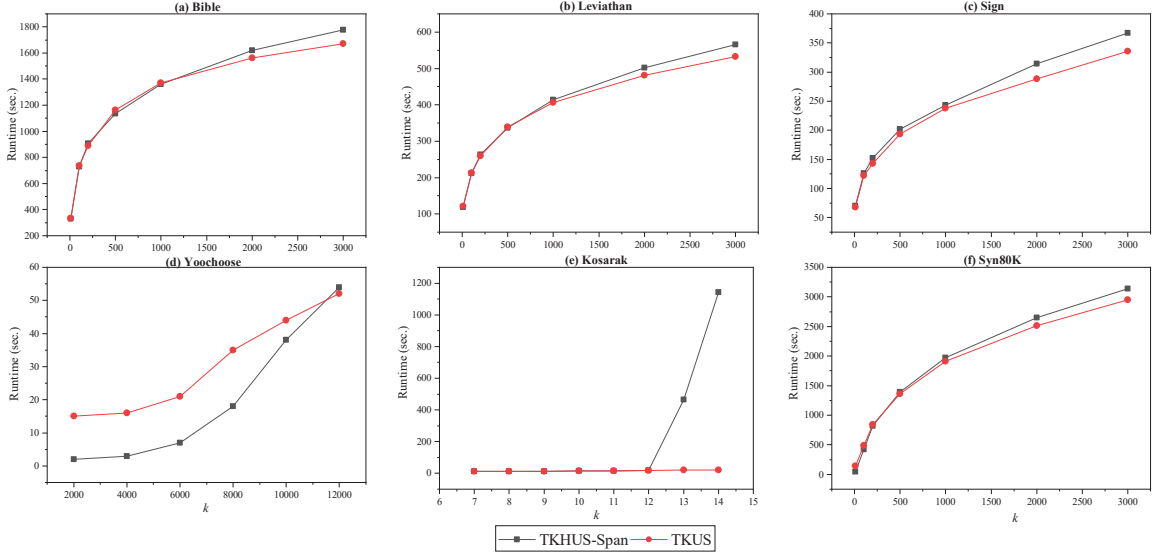}
	\caption{Runtime for various $k$}
	\label{runtime}
\end{figure}

\subsection{Memory Evaluation}

Moreover, TKUS outperforms TKHUS-Span generally in terms of memory usage. The results for the memory usage are illustrated in Figure \ref{memory}. It can clearly be seen that TKUS has better performance than TKHUS-Span in terms of memory usage in most cases. As can be seen from the results on Kosarak and Syn80k, the memory consumed by TKUS is less than that by TKHUS-Span in all instances because the proposed TKUS is able to avoid handling a large number of candidates in the mining process. In addition, there are few cases in the first four datasets where TKUS utilizes more memory, which is probably because TKUS requires more storage space to save the utility and \textit{PEU} values of all 1- and 2-sequences with the SUR strategy. In general, the memory consumed becomes larger with increasing $k$. However, there are some exceptions. Consider the example in Figure \ref{memory}; the memory usage on the Yoochoose dataset is only slightly less than before when setting $k$ = 10000. As another example, on the dataset Sign, the memory consumption sharply decreases as $k$ passes 1000. On the whole, TKUS adopts efficient strategies to significantly improve performance in terms of memory usage. Nonetheless, TKUS still has much room for improvement in terms of memory consumption, especially for large $k$.

\begin{figure}[ht]
	\centering
	\includegraphics[clip,scale=0.7]{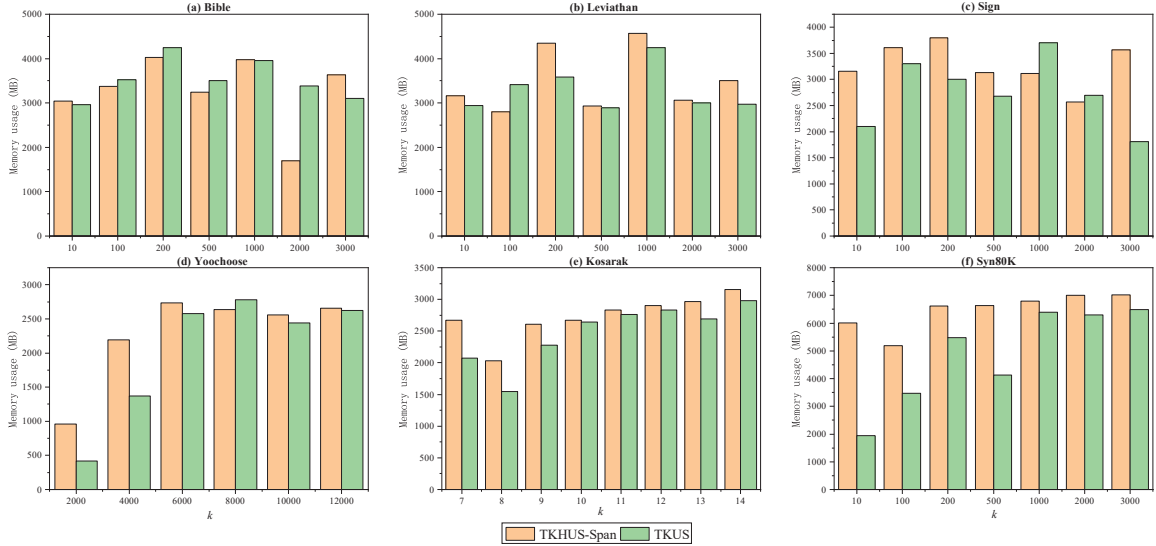}
	\caption{Memory usage for varying $k$}
	\label{memory}
\end{figure}

\subsection{Scalability}

We conducted experiments to evaluate the scalability of TKUS on a synthetic database. We increased its data size through duplication with the purpose of obtaining a series of datasets with different sizes varying from 10K to 120K. The experimental results in terms of execution time and memory with $k$ set to 500 are shown in Figure \ref{scalability}. It is clear that TKHUS-Span has a worse performance than TKUS on the two metrics. As can be seen, the execution time increases linearly as the number of $q$-sequences contained in the databases grows. On the dataset with 120K size, the execution times of the two algorithms are less than those on the dataset with 80K $q$-sequences. In terms of memory usage, there is also a slight decrease when the algorithm addresses the dataset with 120K size because its unique distribution of utility values leads the SUR strategy to increase the minimum utility threshold significantly. From Figure \ref{scalability}, we can conclude that the TKUS algorithm is scalable to large-scale datasets because the running time and memory consumption are almost linearly related to dataset size.

\begin{figure}[ht]
	\centering
	\includegraphics[clip,scale=0.52]{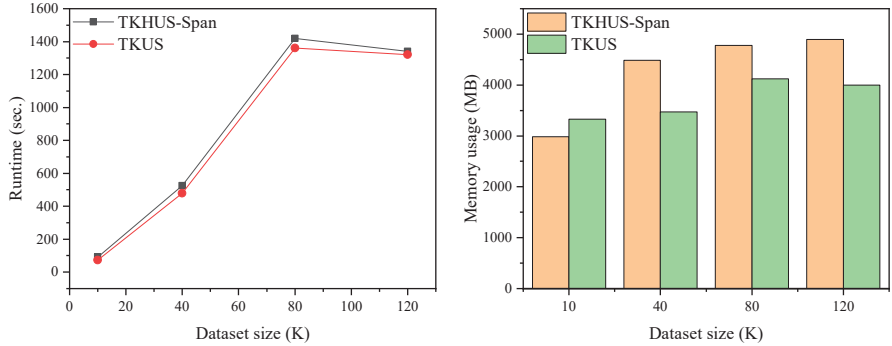}
	\caption{Scalability of the compared methods when $k$ = 500}
	\label{scalability}
\end{figure}

\section{Conclusions}  
\label{sec:conclusion}

In this paper, we first defined some key notations and concepts of SPM and formulated the problem of top-$k$ HUSPM. To handle top-$k$ HUSPM efficiently, we then proposed a novel algorithm called TKUS. For further efficiency improvement, TKUS adopts a projection and local search mechanism and follows several strategies, including SUR, TDE, and EUI, which are able to greatly reduce the search space and speed up the mining process. Finally, we conducted substantial experiments on six datasets with varied characteristics. The experimental results show that TKUS has good performance in terms of execution time, number of candidates, scalability, and more. We can conclude that TKUS is able to increase the minimum utility threshold quickly and extract the top-$k$ HUSPs efficiently. Our future work will apply the proposed TKUS algorithm to big data. For example, several extensions of the TKUS algorithm can be considered, such as adopting parallel techniques, such as Hadoop, to increase its speed.

\section*{Acknowledgment}

Thanks for the anonymous reviewers for their insightful comments, which  improved the quality of this paper. This research was partially supported by the National Natural Science Foundation of China (Grant No. 62002136), Natural Science Foundation of Guangdong Province, China (Grant No. 2020A1515010970), and Shenzhen Research Council (Grant No. GJHZ20180928155209705)

\section*{References}
\bibliographystyle{model1-num-names}
\bibliography{main} 


\end{document}